\documentclass[12pt,american]{article}

\usepackage[T1]{fontenc} 
\usepackage{babel} 
\usepackage{lmodern} 
\usepackage{comment}
\usepackage{tikz}
\usepackage{geometry}
\usepackage{setspace} 
\usepackage[authoryear]{natbib} 

\usepackage{amsmath,amssymb,amsthm,thmtools,mathtools,mathrsfs} 
\usepackage{enumitem,algorithm2e} 

\usepackage{tabularx, booktabs} 

\usepackage[unicode=true,pdfusetitle,
 bookmarks=true,bookmarksnumbered=false,bookmarksopen=false,
 breaklinks=true,pdfborder={0 0 0},pdfborderstyle={},backref=false,colorlinks=true]
 {hyperref}
\hypersetup{linkcolor=blue, citecolor=blue, urlcolor=blue}

\usepackage[normalem]{ulem} 
\usepackage{xcolor} 

\newlength{\lyxlabelwidth}      

\theoremstyle{remark}
\newtheorem{claim}{Claim}
\theoremstyle{plain}
\newtheorem{proposition}{Proposition}
\newtheorem{lemma}{Lemma}
\newtheorem{theorem}{Theorem}

\theoremstyle{definition}
\newtheorem{example}{Example}

\newtheorem{remark}{Remark}

\usepackage{centernot}

\usepackage{amsthm}

\newtheoremstyle{named} 
{} 
{} 
{} 
{} 
{\bfseries} 
{:} 
{.5em} 
{#3} 

\theoremstyle{named}

\let\variant\relax

\declaretheorem[
  name={Property~\variant{\textsuperscript{*}}}, 
  style=definition,
  numbered=no,
]{property*}

\declaretheorem[
  name={Property~\variant{\textsuperscript{\#}}},
  style=definition,
  numbered=no,
]{property-q}

\declaretheorem[
  name={Property~\variant{\textsuperscript{\textdagger}}},
  style=definition,
  numbered=no,
]{property-v}

\newlist{casenv}{enumerate}{4}
\setlist[casenv]{leftmargin=*,align=left,widest={iiii}}
\setlist[casenv,1]{label={{\itshape\ Case} \arabic*.},ref=\arabic*}
\setlist[casenv,2]{label={{\itshape\ Case} \roman*.},ref=\roman*}
\setlist[casenv,3]{label={{\itshape\ Case\ \alph*.}},ref=\alph*}
\setlist[casenv,4]{label={{\itshape\ Case} \arabic*.},ref=\arabic*}

\usepackage[normalem]{ulem} 
\usepackage{xcolor} 

\title{On fairness of multi-center allocation problems\thanks{
We would like to thank {\"O}zg{\"u}n Ekici, Flip Klijn, and the audience at several seminars and conferences for their very helpful comments and discussions. All remaining errors are our own. Yao Cheng is gratefully acknowledge financial support from the National Natural Science Foundation of China (No.72203176), and Feng gratefully acknowledges financial support from the National Natural Science Foundation of China No. 72503028. Refine.ink was used to check the paper for consistency and clarity.}}

\author{
Yao Cheng\thanks{Southwestern University of Finance and Economics, China. Contact:~\protect\href{mailto:chengyao@swufe.edu.cn}{chengyao@swufe.edu.cn}.}
\and
Di Feng\thanks{Dongbei University of Finance and Economics, China. Contact:~\protect\href{mailto:dfeng@dufe.edu.cn}{dfeng@dufe.edu.cn}.}}

\date{\today}

\begin{document}
\maketitle

\begin{abstract}
We investigate \citet{ekici2024}'s multi-center allocation problems, focusing on fairness in this context. We introduce three fairness notions that respect centers' priorities: \emph{internal fairness}, \emph{external fairness}, and \emph{procedural fairness}. The first notion eliminates envy among agents within the same center, the second prohibits envy across different centers, and the third rules out envy from an ex-ante perspective through agents' trading opportunities.

We provide two characterizations of a natural extension of the top-trading-cycles mechanism (TTC) through our fairness notions. 
Precisely, we show that in the presence of \emph{strategy-proofness} and \emph{pair efficiency}, \emph{internal fairness} and \emph{external fairness} together characterize TTC (Theorem~\ref{thm}). Additionally, \emph{strategy-proofness} combined solely with \emph{procedural fairness} also characterizes TTC (Theorem~\ref{thm:pf}).
Furthermore, by adding \emph{internal fairness}, we establish our third TTC characterization, by relaxing \citet{ekici2024}'s \emph{queuewise rationality} to another voluntary participation condition, the \emph{center lower bound} (Theorem~\ref{thm:ir}).
Finally, we define a core solution within this model and characterize it through TTC (Theorem~\ref{thm:PC}).

Our findings offer practical insights for market designers, particularly in contexts such as international cooperation in medical programs and worker exchange programs.
\medskip



\medskip
    

\end{abstract}

\maketitle

\section{Introduction}
We consider allocation problems with incomplete priorities: Agents and objects are partitioned regionally, and priorities are internally induced only within each region. 
The considered model is introduced by \citet{ekici2024}, and it is a generalization of \citet{shapley1974}'s housing markets model and \citet{hylland1979}'s house allocation model. This model
covers important practical applications, e.g., organ transplant network and worker exchange problem.

In organ transplantation networks (e.g., for living donor liver or kidney transplants), collaboration across regions is becoming increasingly important \citep[see, for instance][]{ENCKEP,UNOS1,UNOS2}: While kidney and liver transplantation programs typically operate regionally, cross-program cooperation is necessary for better outcomes, such as helping more patients receive organs and improving match quality.
This kind of collaboration is becoming more common \citep{rees2024p}.
To address this issue, \citet{ekici2024} introduces a new model called \emph{multi-center allocation problems}, focusing on organ transplantation networks.\footnote{
A partial list of papers that study on the context of multi-center kidney exchange and dynamic matching is \citet{roth2007notes,abraham2007clearing,rees2017kidney,agarwal2019market,akbarpour2024unpaired}.} 
In this model, agents (patients) and objects (organs) are partitioned into different medical centers, and each center follows a strict order to rank its members, originally designed for running its own transplantation program locally. 
By accounting for the complex structure of this model, we can incentivize organ transplantation programs to operate more globally, improving efficiency while ensuring voluntary participation and fairness. 

In worker exchange programs, a worker may be replaced by a new one at his home organization. A few examples include the Commonwealth Teacher Exchange Programme \citep{dur2019two}, teacher reassignments in France \citep{combe2022design}, and job rotation in Japan \citep{yu2020market}. In such exchanges, the key question is how to reassign workers in an efficient and fair manner.
The model we consider can also be applied to this type of allocation problems, where employees' assigned positions are switched.
In this context, agents (employees) and objects (job positions) are distributed across different departments within an organization (a firm or a  public sector).
Each department has information only about its own employees, resulting in an internal priority that may reflect the department's evaluation of its employees.

In the context of organ transplantation networks,
\citet{ekici2024} introduces a new voluntary participation condition, \emph{queuewise rationality}, which states that an agent ordered in the $q$-th position in a center is guaranteed to receive an object at least as good as his $q$-th-favorite object in that center.
Furthermore, he demonstrates that a natural extension of the top-trading-cycles mechanism (TTC in short) is characterized by \emph{pair efficiency}, \emph{strategy-proofness}, and \emph{queuewise rationality}. Consequently, one can conclude that TTC is outstanding in this context, as it is efficient, strategically robust, and ensures voluntary participation.
 
\subsection{Overview of the paper}
In this paper, unlike \citet{ekici2024}, we primarily focus on the issue of fairness: Fairness is particularly important for implementation, as participants may reject an unfair mechanism and refuse to use it. 
One relevant example of this happened in school choice in New Orleans, where TTC was used in 2010 but was later replaced by the deferred acceptance mechanism, because TTC was considered unfair in the context of school choice \citep[for details, see][]{abdulkadiroǧlu2020efficiency}.
That is to say, before implementing a mechanism, we need to verify whether it satisfies certain fairness criteria.
To address it, we propose several fairness notions, derived from centers' priorities, aimed at eliminating agents' envies.
\medskip

We begin by proposing two notions that aim to eliminate agents' envies over allocations from an ex-post perspective.
The first one,
\emph{internal fairness}, states that an agent should never envy another agent who is from the same center but has a lower priority than himself.
Next, to determine whether envy between agents from different centers is justifiable, we introduce our second fairness property, \emph{external fairness}, inspired by \citet{morrill2015}.
Following \citet{morrill2015}, if agent $i$ envies agent $j$ from a different center, we check whether reassigning $j$'s allotment to $i$ could potentially harm any agents from $i$'s center who are ranked higher than $i$. If so, we consider $i$'s envy unjustifiable.

The third notion, \emph{procedural fairness}, addresses eliminating agents' envies over trading opportunities from an ex-ante perspective.
We introduce a notable notion, agents' trading opportunity sets: Each agent's trading opportunity set consists of objects belonging to the agent's center, respecting the center's priority for that agent. 
Agents can use objects from their trading opportunity set for direct consumption or for the exchange with other agents.
\emph{procedural fairness} says that 
no set of agents can reach preferred objects by trading only within their trading opportunity sets.
\medskip

Our main findings provide several TTC characterizations through our fairness notions. Specifically, we show that: (1) along with \emph{strategy-proofness} and \emph{pair efficiency}, \emph{internal fairness} and \emph{external fairness} characterize TTC (Theorem~\ref{thm}); and (2) together with \emph{strategy-proofness}, \emph{procedural fairness} also characterize TTC (Theorem~\ref{thm:pf}).
Our results further highlight the exceptional performance of TTC from a fairness perspective.
Moreover, we consider an additional voluntary participation condition, the \emph{center lower bound}, adapted from \citet{chen2025welfare},
which is weaker than \citet{ekici2024}'s \emph{queuewise rationality}. We also provide another TTC characterization by replacing \emph{external fairness} with the \emph{center lower bound} (Theorem~\ref{thm:ir}). 
Finally, we also explore ``(strict) core'' solution in our context.
We modify the standard blocking notion to fit our model and correspondingly have a definition of our core solution, referred to as \emph{proper core}, and show that it is single-valued and coincides with TTC (Theorem~\ref{thm:PC}).

Overall, together with \citet{ekici2024}'s result, our results suggest that TTC perform surprisingly well across four common objectives: Fairness, efficiency, incentives, and voluntary participation.

\subsection*{Organization}
The rest of the paper is organized as follows:
In the next section, we provide an overview of the related literature.
In Section~\ref{sec:model}, we introduce the basic model, the mechanisms with their properties, and a formal description of TTC.
In Section~\ref{sec:result}, we first characterize TTC through our first two fairness properties (Theorem~\ref{thm}). 
We then provide another TTC characterization (Theorem~\ref{thm:pf}) based on our third fairness property.
In Section~\ref{sec:qr}, we distinguish our results from those of \citet{ekici2024} and offer an alternative TTC characterization (Theorem~\ref{thm:ir}) by weakening \citet{ekici2024}'s \emph{queuewise rationality}.
In Section~\ref{sec:stable}, we discuss stability notions within our model, define a core solution in this context, and characterize it via TTC (Theorem~\ref{thm:PC}).
Finally, we conclude the paper in Section~\ref{sec:con} with some final remarks.
Appendix~\ref{appendix:proof} contains all proofs omitted
from the main text, and we demonstrate the independence of our properties in Appendix~\ref{appendix:examples}.

\section{Related literature}
The paper closest to ours is \citet{ekici2024} and we further discuss the relation between his and ours in Section~\ref{sec:qr}.
The main result of this paper contributes to the literature that studies TTC from an axiomatic perspective, e.g., \citet{ma1994}. \citet{morrill2024top} provide a comprehensive survey on this topic.

\subsection*{Fairness}
Our paper is related to a growing literature that investigates 
fairness notions that are compatible with efficiency.
Our first notion, \emph{internal fairness}, is closely related to weak fairness \citep{svensson1994} and respect for priority \citep{coreno2022}, both of which ensure the absence of justified envy. The main difference lies in implications: while their fairness properties are compatible with efficiency, they only result in dictatorial rules, which are often viewed as inequitable. In contrast, our property allows for non-dictatorial mechanisms, highlighting a key distinction.

Our second notion, \emph{external fairness}, aligns with the idea that an agent with envy must propose an alternative allocation. Work in this vein includes concepts such as 
justness \citep{morrill2015}, partial fairness \citep{dur2019school}, essential stability \citep{troyan2020essentially}, and legality \citep{ehlers2020legal}.
Most papers in this literature are loosely based on the idea that priorities can be violated in certain situations, such as when agents ``consent'' to having their priority violated if they cannot benefit from it, while other agents can be better off through the violation of priorities. This reveals a foundational difference between those papers and ours: In their setting, priorities are complete, whereas in ours, priorities are incomplete (recall that each center's priority is only over its own members). As a result, the key challenge in our model lies in defining fairness with incomplete priorities, rather than focusing on priority violations.
Nevertheless, \emph{external fairness} shares similar ideas with these studies, as it also emphasizes the consequences of agents' envies—specifically, how allowing one agent's envy can trigger a chain of reassignments (see Example~\ref{example1}).
While similar in the sense that both consider more than one step ahead, \emph{external fairness} fundamentally differs from notions like essential stability. Essential stability takes an optimistic
approach: When anticipating future steps, they focus on the best possible outcomes, assuming agents achieve their best possible allotments. In contrast, \emph{external fairness} adopts a pessimistic approach, focusing on excluding all potential harms to other agents. More broadly, \emph{external fairness} contributes to the literature on defining stability in scenarios where agents anticipate more than one step of blocking.
The standard concept in this literature is called farsightedness \citep{harsanyi1974equilibrium}: An outcome is farsightedly stable if no series of blocks exists that culminates in better outcomes for every agent participating in it.
This concept is recently explored into priority-based allocation problems, e.g., \citet{dogan2023existence} and \citet{atay2024school}.

Our concept of \emph{procedural fairness} offers new insights by emphasizing potential unfairness in the allocation process rather than in the allocations themselves.  
Relevant literature on procedural fairness includes \citet{KK2006} and \citet{CY2021}. 
Different from ours, \citet{KK2006} focus on random allocations.
They study probabilistic stable mechanisms that meet their requirement: whenever each agent has
the same probability to move at a certain point in the procedure that determines the final probability distribution. 
\citet{CY2021} study priority-based allocation problems and introduce a control system where the control rights of each object are concentrated based on priorities and no coalition has redundant control rights, and propose a non-empty core concept based on the control system. 
This paper differs from their work in the following aspects: 
First, in our setting, priorities are incomplete, whereas in theirs, objects are endowed with complete priorities. Second, in their paper, the control right of each agent relies on that agent's received allotment, while the trading opportunity set of each agent does not.  
Third, unlike their system, which allows for both direct and indirect control, our opportunity set focuses solely on direct control (i.e., consuming your own or doing exchange with other agents) due to the practical challenges of implementing indirect control in kidney exchanges, especially among agents from different centers. 
Another related notion is individual trade stability \citep{papai2013matching}. This notion somewhat is between 
our \emph{external fairness} and \emph{procedural fairness}.
It states that agent $i$'s envy to $j$ is not justified if 
$j$ obtains the object $i$ desires through a trade with a third agent, $k$, who has a higher priority than $i$ for that object.
In other words, individual trade stability eliminates envies
by accounting for agents' tradings rather than depending on objects' priorities.

\subsection*{Incomplete priorities}
More broadly, our paper contributes to the topic of allocation problems with incomplete priorities, e.g., \citet{dur2019two}.
In this type of problems, objects only have rank over partial agents and thus their priorities may be incomplete.
Given this restriction, how to define fairness is crucial. 
In \citet{dur2019two}, their fairness notion, respect for internal priorities is related to objects' eligibility quotas as in their model, objects' capacities may vary. 
However, in our setting, we cannot change the number of objects due to the constraints in reality, e.g., the supply of organs are scarce and may be fixed through a period.
Instead, we focus on fairness within a framework where the size of the problem is fixed.

Our paper also contributes to the growing literature focusing on allocation problems with complex ownership or property rights \citep{ekici2013,BK2019, sun2020, zhang2020, yang2023proper}. In these types of problems, ownership (or agents' rights) can often be translated into partial priorities: owners have strictly higher priority than other agents, while non-owners are treated as indifferent among themselves.
This creates a challenge similar to the one we address: It is hard to compare two agents when there is no strict order to rank them.
In other words, all face the same question: How to expand priorities into linear orders in a fair manner. We further discuss this point in Section~\ref{sec:stable}.

\section{Model}
\label{sec:model}
\subsection{Multi-center allocation problems}
The basic model setup is same to \citet{ekici2024}. There is a finite \emph{set of centers}, $C$. Each center $c\in C$ consists of a finite \emph{set of agents (e.g., patients or workers)} $N_c$ and a finite \emph{set of heterogeneous objects (e.g., organs or job positions)} $O_c$, with $|N_c|=|O_c|$.\footnote{In this paper, we do not consider the outside option because we want to keep balance in our model, we further discuss this point in Section~\ref{sec:con}.}
Each center $c$ has a \emph{strict priority} $\succ_c$ over $N_c$. Each center's priority of its members reflects some form of justice—whether based on need, queue time, or other criteria. However, the design of such priorities is beyond the scope of this paper, and we therefore assume that the priority structure is exogenously given.

\medskip

Let $N\equiv \cup_{c\in C} N_c $ and $O\equiv \cup_{c\in C} O_c $ denote the set of all agents and all objects, respectively. Let $n\equiv |N|=|O|\geq 2$.
For each agent $i\in N$, let $c(i)$ represent the center to which they belong to, i.e., $i\in N_{c(i)}$.
An \emph{allocation} is a bijection $x:N \to O$. Let $X$ denote the set of all allocations.
For each $x\in X$ and each $i\in N$, $x_i\equiv x(i)$ denotes agent~$i$'s \emph{allotment} at $x$.
Let $R = (R_i)_{i \in N}$ be a preference profile over $O$, where $R_i$ denotes the strict preference of agent~$i$. The strict part associated with $R_i$ is denoted by $P_i$.\footnote{That is, for all $a,b \in O$, $a \mathbin{R_i} b$ if and only if $a \mathbin{P_i} b$ or $a = b$.} 
Let $\mathcal{R}$ be the set of all strict preferences. We use the standard notation $(R'_i, R_{-i})$ to denote the profile obtained from $R$ by replacing agent~$i$'s preference relation $R_i$ with $R'_i \in \mathcal{R}$.
\medskip

A \emph{multi-center allocation problem} (or a \emph{problem} in short) is formed by $( (N_c,O_c,\succ_c)_{c\in C},R   )$. As the $(N_c,O_c,\succ_c)_{c\in C}$ remain fixed throughout, we will simply denote the problem by $R$.
Thus, the strict preference profile domain $\mathcal{R}^N$ also denotes the set of all problems.
\medskip

Note that our model extends various classical object allocation models:\footnote{Note that if $n=2$, then our model reduces to either the Shapley-Scarf housing market model or the house allocation model.}
\begin{itemize}
    \item    If for each $c\in C$, $|N_c|=|O_c|=1$, then our model equals the \textbf{Shapley-Scarf housing market model} \citep{shapley1974}; and
    \item if $|C|=1$, then our model is the \textbf{modified house allocation model}. Note that in \citet{hylland1979}'s original model, there is no priority over agents.
\end{itemize}

\subsection{Mechanisms and their properties}
 A \emph{mechanism} is a function $f :   \mathcal{R}^N \to X$ that associates with each problem $R$ an allocation $f(R)$. For each $i \in N$, $f_i(R)$ denotes agent~$i$'s allotment at $f(R)$. 
 \medskip

Here, we introduce several properties of mechanisms that are related to fairness.
\medskip

The first one is \emph{strategy-proofness}, which is most frequently used in the literature on mechanism design.
\bigskip

\noindent \textbf{Strategy-proofness}. For each $R\in \mathcal{R}^N$, each $i\in N$, and each $R'_i\in \mathcal{R}$, $f_i(R) \mathbin{R_i} f_i(R'_i,R_{-i})$. 
\bigskip

Typically, \emph{strategy-proofness} is regarded as an incentive property that eliminates an agent's incentive to misreport his preferences individually. As a result, it is also viewed as a fairness notion, as it ``~`leveling the playing field' idea only indicates that sophisticated students lose their strategic rents under the new mechanism'' \citep{pathak2011}. 
\medskip

Given $R$, we say agent $i$ \emph{envies} another agent $j$ at $x$ if $x_j \mathbin{P_i} x_i$. An ideal allocation should eliminate all agents' envies, resulting in what is known as an \emph{envy-free} allocation. However, since our model only allows deterministic allocations, such ideal situations does not exist in general.\footnote{For instance, if $R_i=R_j$, then for any allocation $x$, it must be the case that either $i$ envies $j$, or $j$ envies $i$. }
Therefore, we consider some weaker properties.
The first one is \emph{pair efficiency} \citep{ekici2022}, which states that no pair of agents should envy each other.
\bigskip

\noindent \textbf{Pair efficiency}. For each $R\in \mathcal{R}^N$, there is no pair of agents $i,j\in N$ such that $f_i(R)\mathbin{P_j} f_j(R)$ and $f_j(R)\mathbin{P_i} f_i(R)$. 
\bigskip

Originally, \emph{pair efficiency} is viewed as an efficiency notion, which is weaker than \emph{Pareto efficiency}.\footnote{An allocation $x$ is \textit{Pareto efficient} if there is no distinct allocation $y\neq x$ such that for each $i\in N$, $y_i\mathbin{R_i} x_i$.}
Here, one may reinterpret it as a fairness notion as it rules out the situation where two agents envy each other, aligning with the idea that swapping two objects eliminates envy \citep{echenique2025swap}.
\medskip

The next property reflects ``a minimal level of mindfulness toward the property rights of individual agents'' \citep{ekici2024}.

For each $c\in C$, and $i\in N_c$, let $r_i=1+|\{j\in N_{c(i)} \mid j\succ_c i\}|$ be his rank according to $\succ_c$, i.e., if $i$ has the $k-$highest priority at $\succ_c$, then $r_i=k$.
Moreover, for each $R\in\mathcal{R}^N$, let $o_i(R)\in O_c$ be his $r_i$-preferred object in $O_c$, according to $R_i$.
Formally, $o_i(R)$ is such that $\{o\in O_c\mid o\mathbin{R_i}o_i(R)\}=r_i$. 
\bigskip

\noindent \textbf{Queuewise rationality}. For each $R\in\mathcal{R}^N$, each $c\in C$, and each $i\in N_C$, $f_i(R)\mathbin{R_i} o_i(R)$.
\bigskip

\subsubsection*{A tale of three fairness properties}

In our model, since we also incorporate (incomplete) priorities over agents, along with the idea of pairwise stability in matching theory, it is natural to consider envy-freeness with respect to these priorities.
However, as we mentioned earlier, the priorities in our model are incomplete: From each center, we have one priority only over its own members. As a result, determining whether an agent’s envy is justified (based on these priorities) becomes challenging due to this incompleteness.\footnote{In Section~\ref{sec:pairwisestable}, we further discuss the classical pairwise stability in our model.}

First, we start with an obvious case: When two agents are from the same center.
In this case, 
the assignment may violate the envious agent's priority since he has a higher priority than the agent who received the object that he desired.
\medskip


Given $f(R)$, agent $i$ \textit{internally justifiedly envies} another agent $j$ at $f(R)$ if 
\begin{itemize}
    \item[(1)] $i$ envies $j$ at $f(R)$;
    \item[(2)] $i$ and $j$ belong to the same center, i.e., $c(i)=c(j)=c$; and 
    \item[(3)] $i\succ_c j$.
\end{itemize}

\noindent \textbf{Internal fairness}. For each $R\in \mathcal{R}^N$, there is no pair of agents $i,j\in N$ such that $i$ \emph{internally justifiedly envies} $j$ at $f(R)$.
\bigskip

\emph{Internal fairness} only eliminates envy that is not ``justified'' 
by a center's priority of its members.
This type of justified fairness, driven by priorities, arises naturally in many applications. For instance, 
priorities guide the allocation of medical resources to individuals based on urgency, determine eligibility for public housing based on queueing, and enable schools to rank students based on exam scores and proximity.
Moreover, \emph{internal fairness} captures incentive issues.\footnote{In the context of school choice with walk-zone rights, \citet{klaus2024minimal} study the similar incentive issues. They refer to it as \textit{minimal-access monotonicity}: No student has an incentive to give up his additional walk-zone rights.}
Consider the case when priorities are induced by a reserve system.
If an agent with higher priority receives a less preferred object than an agent with lower priority, then the higher-priority agent may benefit by manipulating his priority, e.g., by reneging~/~jockeying on it, to obtain a better object. Such strategic actions commonly exist in various reserve systems and have been extensively studied in the fields of queuing theory and operations research (see, e.g., \citet{kullback1997information} and \citet{gross2011fundamentals}).
\emph{Internal fairness} prohibits such strategic actions, as all agents prefer higher priorities.

\begin{remark}{\textbf{Weak internal fairness}}\ \\
One might argue that the case above is not fully justified, since $i$ may not have a claim to receive $j$'s allotment if it is outside their common center $c$.
Thus, we also consider a weak version of \emph{internal fairness}, which additionally requires that $j$'s allotment is also belongs to center $c$.

\noindent \textbf{Weak internal fairness}. For each $R\in \mathcal{R}^N$, there is no pair of agents $i,j\in N$ such that
(1) $f_j(R)\in O_{c(i)}$, and (2) $i$ \emph{internally justifiedly envies} $j$ at $f(R)$. 

Alternatively, one can restate it with the following: For each $c\in C$, and $o\in O_c$, let $\succ_o=\succ_c$. We say agent $i$ has higher priority than  $j$ at an object $o$ if $i\succ_o j$.
Then, \emph{weak internal fairness} says that there is no pair of agents $i,j\in N$ such that $i$ envies $j$ at $f(R)$, and $i$ has higher priority than $j$ at $f_j(R)$. \hfill $\diamond$ 
\label{remark:wif}
\end{remark}

By \emph{internal fairness}, we rule out envy among agents within the same center based on that center's priority. However, envy may still arise between agents from different centers, where no priority can be directly applied to determine whether such envy is ``justified.'' Nonetheless, as pointed out by \citet{yang2023proper}, there may be an indirect relationship derived from priorities that can be applied to eliminate such envy. To address this, we introduce our second fairness property, \emph{external fairness}, which is a modification of Morrill's \emph{justness} \citep{morrill2015}.
The intuitive interpretation of \emph{external fairness} is as follows.
If $i$ envies another agent $j$ from a different center and raises a complaint about the selected allocation, we must determine whether the complaint is justifiable.\footnote{
Note that here we use the term ``complaint,'' which differs from the term ``violation of priority'' we used earlier to define \emph{internal fairness}. The fundamental distinction between the two arises from the difficulty we face: in the case of a priority violation, we can directly use a center's priority to justify an agent's envy. In contrast, for a complaint, no such priority can be directly applied. Thus, violations of priorities do not exist, making it harder to determine whether the envy is justified.}
Since $i$ and $j$ belong to different centers, no exogenously given priority can be directly applied to judge $i$'s complaint. However, as Morrill pointed out, we can consider a hypothetical scenario: If we accept this complaint and reassign $j$'s allotment, what would be the consequence? Specifically, would another agent $k$, who belongs to the same center as $i$ and has a higher priority, be directly harmed by this reassignment?\footnote{Here, by ``directly harmed'' we mean that $k$'s allotment is directly reassigned to $j$, which aligns with the idea of \citet{morrill2015}. In other words, we formalize harm by considering only immediate effects. More complex chains, in which $j$’s deviation indirectly worsens $k$’s assignment without $j$ ever receiving $k$'s allotment, are not considered, as the current notion is sufficient to formalize our fairness concept and guarantee positive results.} If so, agent $k$, with a higher priority, may veto $i$'s complaint. Therefore, we should reject $i$'s complaint, as we prioritize $k$'s welfare over $i$'s, according to the priority of the center they both belong to.
\medskip

This idea can also be interpreted in a manner similar to the concept of farsighted stability \citep{harsanyi1974equilibrium}, which considers the final allocation that may result from admitting certain envies.
It is also applied in matching theory to investigate weaker definitions of stability, e.g., legality \citep{ehlers2020legal} and
essential stability \citep{troyan2020essentially}. 

We now present an example to highlight the key insight behind our idea.

\begin{example}
\label{example1}
There are two centers $c$ and $c'$, along with three agents and three objects as follows. $N_c=\{1,2\},O_c=\{a,b\}$, $N_{c'}=\{3\}$, and $O_{c'}=\{d\}$.  Agents' preferences and centers' priorities are given as below (agents'
``endowed objects'' are underlined for emphasis). 

$$R_1:\underline{a},\underline{b},d;$$
$$R_2:d,\underline{a},\underline{b};$$
$$R_3:a, \underline{d},b;$$
$$\succ_c:2,1;\text{ and }\succ_{c'}:3.$$

\begin{center}
\begin{tabular}{|c|c|l|l|}
\hline
    Allocation & $1$ & $2$ & $3$ \\\hline
    $x$ & $b$  & $d$ & $a$ \\\hline
    $y$ & $b$  & $a$ & $d$ \\\hline
\end{tabular}
\end{center}

Consider the allocation $x$.
We see that only agent $1$ envies another agent: $x_3=a \mathbin{P_1}b= x_1$. Agent $1$ may argue that his envy is reasonable since $a$ belongs to center $c$, the center to which agent $1$ belongs to (i.e., $1\in N_c$ and $b\in O_c$).

If we admit this envy, then agent $3$ cannot receive $a$. Consequently, agent $3$ must receive $d$; otherwise, agent $3$ would envy whoever receives $d$, and by the same logic used to justify agent $1$'s envy, we must also admit agent $3$'s envy, as $d$ belongs to $c'$, the center to which agent $3$ belongs to.
As a result, agent $2$ cannot receive $d$. In other words, agent $2$'s allotment at $x$ ``depends'' on agent $3$'s reported preference, in the sense of \citet{morrill2015}. 
Therefore, $\{a, b\}$ must be assigned to agents $\{1, 2\}$. Since agent $2 \succ_c 1$, it is reasonable to assign $a$ to agent $2$. 

In summary, if we admit agent $1$'s envy at $x$, the resulting final allocation is $y$. So, agent $1$ does not benefit from his envy as $x_1=y_1=b$, and everyone else ends up worse off. 
This means that $1$'s envy is ``vacuous'' in the sense of \citet{troyan2020essentially}, and suggests that admitting $1$'s envy at $x$ is ``farsightedly unstable.'' 
\hfill $\diamond$  
\end{example}

To formalize our idea, we introduce Morrill's notion of dependence, which captures potential harm arising from reassignments. 
An agent $i$ \textit{depends} on another agent $j$ at $f(R)$ if there exists $R'_j$ such that $f_j(R'_j,R_{-j})=f_i(R)$.

Given $f(R)$, $i$ \textit{externally justifiedly envies} $j$ at $f(R)$ if 
\begin{itemize}
    \item[(1)] $i$ envies $j$ at $f(R)$; 
    \item[(2)] $i$ and $f_j(R)$ belong to the same center, while $j$ does not, i.e., $c(i)\neq c(j)$ and $f_j(R)\in O_{c(i)}$; and
    \item[(3)] for each agent $k\in  \{\ell\in N_{c(i)}\mid  \ell\succ_{c(i)} i \} $, $k$ does not \emph{depend} on $j$ at $f(R)$.\footnote{If $\{\ell\in N_{c(i)}\mid  \ell\succ_{c(i)} i \}$ is empty, then this condition is vacuously satisfied. \label{footnote:vacous}}
\end{itemize}

\noindent \textbf{External fairness}. For each $R\in \mathcal{R}^N$, there is no pair of agents $i,j\in N$ such that $i$ \emph{externally justifiedly envies} $j$ at $f(R)$.
\bigskip

\begin{remark}{\textbf{Strong external fairness}}\ \\
Unlike \emph{internal fairness} and \emph{weak internal fairness} (see Remark~\ref{remark:wif}), here we cannot drop the requirement that $i$ and $f_j(R)$ belong to the same center. 
Otherwise, any cross-center envy that a highest-priority agent $i$ has toward any $j(\not\in N_{c(i)})$ would be classified as externally justified (see Footnote~\ref{footnote:vacous}), leading to an impossibility result. In other words, no mechanism can satisfy this stronger version of \emph{external fairness}.
\hfill $\diamond$ 
\end{remark}

By \emph{external fairness}, we first check whether reassigning $j$'s allotment could potentially harm any agents who have higher priority than agent $i$. If so, we treat $i$'s complaint unjustifiable and reject it.
In other words, \emph{external fairness} serves to rule out envies
that could potentially harm another agent with higher priority.
This aligns well with common sense of fairness, similar to how people generally accept that those with higher priority on a waiting list or in a queue should not be displaced.
\medskip

Finally, we discuss our last fairness property, which takes an ex-ante perspective. That is, it concerns the fairness of the process itself, based on agents' ``trading opportunities.''

Agents' trading opportunities can be understood as follows. Any allocation $x$ can be decomposed into several trading cycles (these cycles are not necessarily unique).\footnote{Loosely speaking, a trading cycle is an ordered set of agents and objects that describes who trades with whom, and for what. } 
Thus, $x$ can be viewed as being obtained through the following process: Each agent $i$ uses an object $o\in O_{c(i)}$ to trade with other agents for $x_i$ (if $x_i\in O_{c(i)}$, then 
$i$ effectively trades with himself). In this sense, each agent $i$ utilizes his trading opportunities to obtain $x_i$.

The next question is where these trading opportunities come from and how we should evaluate them to ensure fairness.
Loosely speaking, the trading opportunities available to an agent may come from resources that he controls and that he can consume/utilize directly.
Ideally, agents' trading opportunities should be equal when they are treated equally ex-ante.
However, this is not the case in our model: Here, agents within the same center are already treated unequally due to that center's priority.
Thus, in our context, a just distribution of agents' trading opportunities must also respect centers' priorities. Specifically, if two agents come from the same center, the agent with higher priority should have more opportunities (or, a weakly larger opportunity set) than the agent with lower priority.
To formalize this idea, we introduce the concept of agents' trading opportunity sets by adopting the algorithm of \citet{papai2000}.
In the algorithm below, at the beginning, each agent with the highest priority ``controls'' all objects that belong to his center, meaning he has the right to use any object within that set. 
Then, we let each of them points to his most preferred object, and let each object points to the highest priority agent in its center.
Once these agents and their corresponding trading objects are identified, the same process is applied to the remaining agents and objects. Specifically, each agent with the highest priority among the remaining agents ``controls'' all remaining objects that belong to his center.
This iterative procedure continues, determining each agent's trading opportunities given the preference profile.

\noindent \rule{1\columnwidth}{1pt}
\paragraph{Construction of trading opportunity sets($\mathcal{T}$)} 
\begin{description}
\item[Step~$0$]: Let $N^1=N$, and $O^1=O$. 
    \item[Step~$\tau$ $(\geq 1)$:] For each $c\in C$ with $N_c^\tau\neq \emptyset$, let $i_c^\tau \equiv \max_{\succ_c} N_c^\tau$ be the agent with the highest priority among the remaining agents in $N_c^\tau$.
    For each $c\in C$, let $i_c^\tau$ point to his best preferred object among remaining objects according to $R_{i_c^\tau}$; and let each remaining object in $O_c^\tau$ point to $i_c^\tau $. Define the trading opportunity set $T_{i_c^\tau}(R)$ as the set of objects pointing to $i_c^\tau$.
    Also, given agents and objects' pointing graph, there exists at least one \emph{cycle}. Remove all agents and objects involved in cycles. Let $N^{\tau+1},O^{\tau+1}$ be the remaining agents and objects, respectively. If $N^{\tau+1}$ is not empty, then proceed to step~$\tau+1$.
    \item[Termination:] The algorithm terminates when no agent remains.
\end{description}
\noindent \rule{1\columnwidth}{1pt}
\medskip

We use $\mathcal{T}$ to represent the algorithm above. 
The algorithm naturally defines each agent's trading opportunity set because it captures a practical implementation as follows: Agents are asked, in the order of their priorities, which available object they would use for consumption or trade, to best meet their preferences. 
For each agent with the highest priority, his trading opportunity set equals the set of objects in their center. 
That is, the opportunity set for each top-ranked agent does not depend on any agent's preference. 
Moreover, for each remaining agent $j$, the trading opportunity set $T_j(R)$ does not depend on $j$'s preference $R_j$, but on the preferences of agents who have higher priority than $j$.  

Now, we present our idea of fairness in the process. 
Since the selected allocation is produced by agents either consuming or trading  objects from their trading opportunity sets, unsatisfied agents have incentives to deviate from the selected allocation if they can consume or trade objects from their trading opportunity sets to achieve higher welfare. We consider these two possibilities. 

Let $x\in X$ be the selected allocation.
First, direct consumption. Suppose that an agent $i$ is unsatisfied at $x$ because he envies another agent $j$'s allotment $x_j$. If $x_j\in T_i(R)$, then it means that $i$ can directly consume $x_j$ from his trading opportunity set.
Second, trading. 
Suppose that a set of agents $S$ who from different centers realize they can improve their welfare by trading objects from their respective trading opportunity sets. In this case, all agents in $S$ can deviate from $x$, by using the objects from their trading opportunity sets and proceeding with the exchange. When cooperation across multiple centers is allowed, such trading represents the most easily achievable form of exchange.
\medskip

Formally, an allocation $x$ is \textit{unfairly produced} at $R$ if 
there exist a list of centers $c_1,\ldots,c_L\in C$, a list of agents $(i_1,\ldots,i_L)\in \Pi_{\ell=1}^L N_{c_\ell}$, and a list of objects $(o_1,\ldots,o_L)\in \Pi_{\ell=1}^L T_{i_\ell}(R)$ such that for each $i_\ell\in \{i_1,\ldots,i_L\}$, $o_{i_\ell +1}\mathbin{P_{i_\ell}}x_{i_\ell}$ (mod $L$). 
Otherwise $x$ is \textit{fairly produced} at $R$.
Note that if there is only one center (so $L=1$), then the condition represents the direct consumption, otherwise it represents the trading among agents $i_1,\ldots,i_L$.

\medskip

\noindent\textbf{Procedural fairness}. For each $R\in \mathcal{R}^N$, $f(R)$ is \emph{fairly produced} at $R$.
\medskip

\emph{Procedural fairness} offers a fresh perspective for evaluating mechanisms. For instance, while \citet{papai2000} hierarchical exchange mechanisms satisfy \emph{procedural fairness}, \citet{pycia2017}'s trading cycles mechanisms do not.\footnote{In \citet{pycia2017}'s trading cycles mechanisms, brokers cannot receive the objects they control, which results in a violation of \emph{procedural fairness}.}
This suggests the value of \emph{procedural fairness} in distinguishing mechanisms.
Furthermore, since trading cycles mechanisms are \emph{Pareto efficient}, its violation of \emph{procedural fairness} also implies that \emph{procedural fairness} is logically independent of \emph{Pareto efficiency} (and \emph{pair efficiency}).\footnote{See Appendix~\ref{appendix:examples} for independence examples.}
\medskip

Note that \emph{procedural fairness} implies \emph{queuewise rationality}, but not vice versa.

\begin{lemma}
All \emph{procedurally fair} mechanisms satisfy \emph{queuewise rationality}; however, the opposite is not true.\label{lemma:pfqr}
\end{lemma}

\begin{remark}{\textbf{Relation to other properties}}\ \\
        If for each $c\in C$, $|N_c|=|O_c|=1$, i.e., reduced to the Shapley-Scarf housing market model, then \emph{external fairness} coincides with \emph{individual rationality}, and \emph{procedural fairness} coincides with \emph{weak core stability}; and if $|C|=1$, i.e., reduced to the house allocation model, then \emph{internal fairness}, \emph{procedural fairness}, and \citet{svensson1994}'s~\emph{weak fairness} all coincide with one another.\footnote{Weak fairness states that an agent should never envy another agent who has a lower priority than himself.} Also, in Appendix~\ref{appendix:examples}, we provide several examples to show that our three fairness properties are logically independent for general multi-center allocation problems. \hfill $\diamond$ 
\end{remark}

\subsection{Top Trading Cycles}
For each problem $R$, the top trading cycles mechanism (TTC) selects the allocation determined by the following \emph{top trading cycles algorithm} at $R$, which we call $TTC(R)$.

\noindent \rule{1\columnwidth}{1pt}
\paragraph{Top Trading Cycles Algorithm:} 
\begin{description}
\item[Step~$0$]: For each $c\in C$ and each $o\in O_c$, let $\succ_o=\succ_c$ be the associated priority of $o$.  Let $\succ_O=(\succ_o)_{o\in O}$ be objects' priority profile.
    \item[Step~$\tau$ $(\geq 1)$:] Each remaining agent points to his most-preferred remaining object given $R$. 
    Each remaining object points to the agent with the highest priority among the remaining agents, according to $\succ_O$.\footnote{Note that, since for each center $c$ the balance constraint $|N_c|=|O_c|$ holds, if an object in $O_c$ remains at step~$t$, then some agent in $N_C$ must also remain at step~$t$. As a result, each object in $O_c$ points only to agents in $N_c$.}
    There exists at least one \emph{cycle}. \emph{Execute} all cycles by assigning each agent involved in a cycle the object to which he points. Remove all agents and objects involved in a cycle. If some objects remain, then proceed to step~$\tau+1$.
    \item[Termination:] The algorithm terminates (in at most $n$ steps) when no agent~/~object remains.
\end{description}
\noindent \rule{1\columnwidth}{1pt}

TTC satisfies all of our fairness properties.

\begin{proposition}
 TTC satisfies \textbf{internal fairness} (and hence \textbf{weak internal fairness}), \textbf{external fairness}, and \textbf{procedural fairness}. \label{proposition:sufficicent}
\end{proposition}

\section{Main characterizations}
\label{sec:result}

\subsection{Internal fairness and external fairness}

We characterize TTC through the first four properties that we presented previously. 
\begin{theorem}
    A mechanism $f$ satisfies \textbf{strategy-proofness}, \textbf{pair efficiency},
    \textbf{internal fairness}, and \textbf{external fairness}, if and only if $f=TTC$. \label{thm}
\end{theorem}

Note that from the proof of Theorem~\ref{thm} below, one can see that our characterization in Theorem~\ref{thm} also holds if \emph{internal fairness} is weakened to \emph{weak internal fairness} (see Remark~\ref{remark:wif}).

\subsubsection{Proof of Theorem~\ref{thm}}
\label{sec:proofthm1}
By \citet{pycia2017} and Proposition~\ref{proposition:sufficicent}, 
we know that TTC satisfies the four properties.

Thus, we only need to prove the \textbf{uniqueness part}. 
Let $f$ be \emph{strategy-proof}, \emph{pair efficient}, \emph{internally fair}, and \emph{externally fair}. 
The proof below is by contradiction.
In short, assuming $f\neq TTC$, we use \citet{ekici2022}'s notion of similarity to select a minimal preference profile on which the two mechanisms differ, and then construct a sequence of preference manipulations by agents in a non-executed cycle to derive a contradiction with efficiency of $f$.

For each $\tilde{R}\in \mathcal{R}^N$, and $\tau=1,\ldots,n$, let $\mathcal{C}_\tau(\tilde{R})$ be the set of cycles that are obtained at step~$\tau$ of $TTC(\tilde{R})$, and let $\mathcal{C}(\tilde{R})$ be the set of all trading cycles that are obtained via $TTC(\tilde{R})$. 
For each cycle $\mathcal{C}\in \mathcal{C}(\tilde{R})$, let $S_\mathcal{C}$ be the agents that are involved in $\mathcal{C}$. We say that $\mathcal{C}$ is executed at $f(\tilde{R})$, if for each $i \in S_\mathcal{C}$, $f_{i}(\tilde{R})=TTC_{i}(\tilde{R})$.
\medskip

By contradiction, suppose that $f\neq TTC$. 
We start by selecting a profile which is ``minimal'' according to \citet{ekici2022}'s ``similarity.''
The similarity between $f$ and $TTC$ is a function $\rho$ such that for each $\tilde{R}\in \mathcal{R}^N$, $\rho(\tilde{R})$ is defined as below.
\begin{itemize}
	\item At step~$\tau(\geq 1)$ in $TTC(\tilde{R})$. If there exist cycle(s) in $\mathcal{C}_{\tau}(\tilde{R})$ that are not executed at $f(\tilde{R})$, let $\hat{C}\in \mathcal{C}_{\tau}(\tilde{R})$ be the cycle that has the smallest number of involved agents among all non-executed cycles at $f(\tilde{R})$, then $\rho(\tilde{R})=(\tau,|S_{\hat{C}}|)$; otherwise proceed to step~$\tau+1$. If all cycles in $\mathcal{C}(\tilde{R})$ are executed, then $\rho(\tilde{R})=(n+1,n+1)$.
\end{itemize}

For each pair of similarity $(x_1,x_2),(y_1,y_2)\in \mathbb{N}^2$, we say $(x_1,x_2)< (y_1,y_2)$ if $x_1<y_1$ or [$x_1=y_1$ and $x_2<y_2$]; and $(x_1,x_2)\leq (y_1,y_2)$ if $(x_1,x_2)< (y_1,y_2)$ or $(x_1,x_2)= (y_1,y_2)$. 
We also define $>$ and $\geq$ similarly.
\medskip

Let $R\in \mathcal{R}^N$ be such that for each $\tilde{R}\in \mathcal{R}^N$, $\rho(R)\leq \rho(\tilde{R})$. Let $t=(t_1,t_2)\equiv \rho(R)$ and $x\equiv TTC(R)$.\footnote{Note that for each $R\in \mathcal{R}^N$, $\rho(R)\leq (n+1,n+1)$, and $\rho(R)= (n+1,n+1)$ if and only if $f(R)=TTC(R)$.}

Thus, $f(R)$ executes all cycles in $\cup_{\tau=1}^{t_1-1} \mathcal{C}_\tau(R)$, but it does not execute some cycle in $\mathcal{C}_{t_1}(R)$. Let $\mathcal{C}\in \mathcal{C}_{t_1}(R)$ be the cycle that is not executed with $|S_\mathcal{C}|=t_2$.
Let $O^{t_1}$ be the set of objects that are remaining at step~$t_1$ under $TTC(R)$.
By the definition of TTC, we know that for each $i\in S_C$, $x_i$ is $i$'s most preferred object in $O^{t_1}$, i.e., for each $o\in O^{t_1}$, $x_i\mathbin{R_i} o$.

\begin{claim}
    $|S_\mathcal{C}|=t_2\geq 2$. \label{claim1}
\end{claim}
\begin{proof}[Proof of Claim~\ref{claim1}]
Suppose that $|S_\mathcal{C}|=1$, hence it is without loss of generality to assume that $S_\mathcal{C}=\{i\}$ and $\mathcal{C}=(i\to x_i\to i)$. It implies that $x_i\in O_{c(i)}$.
Since $f(R)$ does not execute $\mathcal{C}$, there is an agent $j\neq i$ such that $f_j(R)=x_i$.
Hence $f_j(R)\neq x_j$, this implies that $j$ does not receive $x_j$
before step~$t_1$ under $TTC(R)$.
Since $f_j(R)$ is $i$'s most preferred object in $O^{t_1}$, $f_j(R) \mathbin{P_i}f_i(R)$, i.e., $i$ envies $j$ at $f(R)$.
There are two cases. 

\textbf{Case~One}: $c(j)=c(i)$. In other words, $j\in N_{c(i)}$.
By the selection of $R$, $f(R)$ executes all cycles in $\cup_{\tau=1}^{t_1-1} \mathcal{C}_\tau(R)$. 
By the definition of TTC, each  $k\in N_{c(i)}$ with higher priority than $i$ receives $x_k$ from some cycle in $\cup_{\tau=1}^{t_1-1} \mathcal{C}_\tau(R)$.
Thus, for each agent $k\in N_{c(i)}$ with $k\mathbin{\succ_{c(i)}} i$, $f_k(R)=x_k$. Therefore, $i\mathbin{\succ_{c(i)}} j$.
So, \emph{(weak) internal fairness} is violated.

\textbf{Case~Two}: $c(j)\neq c(i)$, and hence $j\not \in N_{c(i)}$.
By the selection of $R$, for each $R'_j\in\mathcal{R}$, $\rho(R'_j,R_{-j})\geq t=\rho(R)$. Thus, $f(R'_j,R_{-j})$ executes all cycles in $\cup_{\tau=1}^{t_1-1} \mathcal{C}_\tau(R'_j,R_{-j})$. 
Recall that $j$ does not receive $x_j$ before step~$t_1$ under $TTC(R)$.
Thus, by the definition of TTC, $\cup_{\tau=1}^{t_1-1} \mathcal{C}_\tau(R'_j,R_{-j})=\cup_{\tau=1}^{t_1-1}\mathcal{C}_\tau(R)$.
Hence, for each agent $k\in N_{c(i)}$ with $k\mathbin{\succ_{c(i)}} i$, $f_k(R'_j,R_{-j})=x_k$. Thus, 
$k$ does not depend on $j$.
Also recall that $i\in N_{c(i)}$, $f_j(R)=x_i\in O_{c(i)}$, and $j\not\in N_{c(i)}$.
Thus, $i$ \emph{externally justified envies} $j$ at $f(R)$, contradicts \emph{external fairness} of $f$. 
\end{proof}
 
Note that our two fairness notions are only used to prove Claim~\ref{claim1}.

By Claim~\ref{claim1}, $|S_\mathcal{C}|=t_2=K\geq 2$. Without loss of generality, let $\mathcal{C}=(i_1\to x_{i_1}\to\ldots\to i_K\to x_{i_K}\to i_1)$, and assume that $f_{i_1}(R)\neq x_{i_1}$. Let $c\equiv c(i_1)$.
That is, $i_1\in N_c$ and hence $x_{i_K}\in O_c$. Also, since $K\geq 2$, $x_{i_1}\not \in O_c$.
 
Let $R'_{i_1}:x_{i_1},x_{i_K},\ldots$
By \emph{strategy-proofness} of $f$, $f_{i_1}(R'_{i_1},R_{-i_1})\neq x_{i_1}$; otherwise agent $i_1$ has an incentive to misreport $R'_{i_1}$ at $f(R)$.
Note that if $i_1$ change his report from $R'_{i_1}$ to $R''_{i_1}$ by shifting $x_{i_K}$ to the top, then we have a cycle $(i_1\to x_{i_K} \to i_1)$ in $\cup^{t_1}_{\tau=1} \mathcal{C}_\tau(R''_{i_1},R_{-i})$. 
This is because that the last higher-priority agent in $c(i_1)$ is removed 
before step~$t_1-1$ under $TTC(R''_{i_1},R_{-i})$, so $(i_1\to x_{i_K} \to i_1)$ appears no later than step~$t_1$.
Then, by applying similar arguments we used to prove Claim~\ref{claim1}, we see that this cycle is executed at $f(R''_{i_1},R_{-i_1})$, i.e., $f_{i_1}(R''_{i_1},R_{-i_1})=x_{i_K}$.
Thus, by \emph{strategy-proofness} of $f$,
we conclude that $f_{i_1}(R'_{i_1},R_{-i_1})=x_{i_K}$; otherwise $i_1$ has an incentive to misreport $R''_{i_1}$ at $f(R'_{i_1},R_{-i_1})$.
Therefore, $f_{i_K}(R'_{i_1},R_{-i_1})\neq x_{i_K}$.

Let $R'_{i_K}: x_{i_K},x_{i_1},\ldots $ and $R'\equiv (R'_{i_1},R'_{i_K},R_{N\setminus\{i_1,i_K\}})$.
Since $f$ is \emph{strategy-proof}, $f_{i_K}(R'_{i_1},R_{-i_1})\neq x_{i_K}$ implies that $f_{i_K}(R')\neq x_{i_K}$. 

By dropping $x_{i_K}$ from top to the bottom, we obtain  $R''_{i_K}:x_{i_1},\ldots x_{i_K}$ such that for any two objects $o,o'\in O\setminus \{x_{i_K}\}$, $o \mathbin{R''_{i_K}} o' $ if and only if $o \mathbin{R'_{i_K}} o' $. Thus, by \emph{strategy-proofness} of $f$, $f_{i_K}(R')=f_{i_K}(R''_{i_K},R'_{-i_K})$.

By the definition of TTC, there is one cycle $(i_K\to x_{i_1} \to \ldots \to i_K)=\mathcal{C}'\in \cup_{\tau=1}^{t_1} \mathcal{C}_\tau(R''_{i_K},R'_{-i_K})$ such that $S_{\mathcal{C}'}=S_\mathcal{C}\setminus \{i_1\}$ and hence $|S_{\mathcal{C}'}|=t_2-1=K-1$,
See the figure below for the graphical explanation.

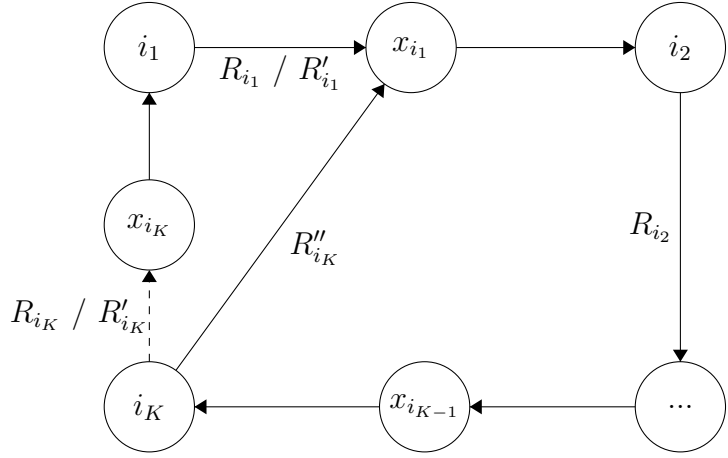
\begin{center}
    \begin{tikzpicture}[scale=0.2]
    \tikzstyle{every node}+=[inner sep=0pt]
    \draw [black] (11.4,-13.6) circle (3);
    \draw (11.4,-13.6) node {$i_1$};
    \draw [black] (28.8,-13.6) circle (3);
    \draw (28.8,-13.6) node {$x_{i_1}$};
    \draw [black] (11.4,-37.5) circle (3);
    \draw (11.4,-37.5) node {$i_K$};
    \draw [black] (29.7,-37.5) circle (3);
    \draw (29.7,-37.5) node {$x_{i_{K-1}}$};
    \draw [black] (46.7,-37.5) circle (3);
    \draw (46.7,-37.5) node {$...$};
    \draw [black] (11.4,-25.4) circle (3);
    \draw (11.4,-25.4) node {$x_{i_K}$};
    \draw [black] (46.7,-13.6) circle (3);
    \draw (46.7,-13.6) node {$i_2$};
    \draw [black] (31.8,-13.6) -- (43.7,-13.6);
    \fill [black] (43.7,-13.6) -- (42.9,-13.1) -- (42.9,-14.1);
    \draw [black] (43.7,-37.5) -- (32.7,-37.5);
    \fill [black] (32.7,-37.5) -- (33.5,-38) -- (33.5,-37);
    \draw [black] (26.7,-37.5) -- (14.4,-37.5);
    \fill [black] (14.4,-37.5) -- (15.2,-38) -- (15.2,-37);
    \draw [black] (11.4,-22.4) -- (11.4,-16.6);
    \fill [black] (11.4,-16.6) -- (10.9,-17.4) -- (11.9,-17.4);
    \draw [black] (14.4,-13.6) -- (25.8,-13.6);
    \fill [black] (25.8,-13.6) -- (25,-13.1) -- (25,-14.1);
    \draw (20.1,-14.1) node [below] {$R_{i_1}~/~R'_{i_1}$};
    \draw [black] (13.17,-35.07) -- (27.03,-16.03);
    \fill [black] (27.03,-16.03) -- (26.16,-16.38) -- (26.97,-16.97);
    \draw (20.69,-26.93) node [right] {$R''_{i_K}$};
    \draw [dashed] (11.4,-34.5) -- (11.4,-28.4);
    \fill [black] (11.4,-28.4) -- (10.9,-29.2) -- (11.9,-29.2);
    \draw (11.4,-31.45) node [left] {$R_{i_K}~/~R'_{i_K}$};
    \draw [black] (46.7,-16.6) -- (46.7,-34.5);
    \fill [black] (46.7,-34.5) -- (47.2,-33.7) -- (46.2,-33.7);
    \draw (46.2,-25.55) node [left] {$R_{i_2}$};
    \end{tikzpicture}
    \end{center}

By the selection of $R$, $\rho(R''_{i_K},R'_{-i_K})\geq (t_1,t_2)=\rho(R)$. Thus, $f(R''_{i_K},R'_{-i_K})$ executes $\mathcal{C}'$. Hence, $f_{i_K}(R''_{i_K},R'_{-i_K})= x_{i_1}=f_{i_K}(R')$. Since $f_{i_K}(R')= x_{i_1}$, $f_{i_1}(R')\neq  x_{i_1}$. Again, by applying similar arguments we used above, we conclude that $f_{i_1}(R')=x_{i_K}$. 
Given $f_{i_K}(R') =x_{i_1}$ and $f_{i_1}(R')=x_{i_K}$, we see that $f_{i_K}(R')  \mathbin{R'_{i_1}}  f_{i_1}(R')$ and $f_{i_1}(R') \mathbin{R'_{i_K}} f_{i_K}(R')$, contradicts \emph{pair efficiency} of $f$. 
This desired contradiction completes the proof of Theorem~\ref{thm}.

\subsection{Procedural fairness}
Next, we characterize TTC through \emph{procedural fairness}.

\begin{theorem}
\label{thm:pf}
A mechanism $f$ satisfies \textbf{strategy-proofness} and \textbf{procedural fairness} if and only if $f=TTC$. 
\end{theorem}

\subsubsection{Proof of Theorem~\ref{thm:pf}}

By \citet{pycia2017} and Proposition~\ref{proposition:sufficicent}, 
we know that TTC satisfies the two properties.
Thus, we only need to prove the \textbf{uniqueness part}. 
Here, we apply \citet{Sethuraman2016}'s approach.
Let $f$ be \emph{strategy-proof} and \emph{procedurally fair}.

We first introduce some notation. 
For each problem $R\in \mathcal{R}^N$, recall that $\mathcal{C}(R)=\cup_\tau \mathcal{C}_\tau(R)$ is the set of all trading cycles that are obtained via $TTC(R)$. For each $\mathcal{C}\in \mathcal{C}(R)$, and each involved agent $i\in S_{\mathcal{C}}$, let $e_i(R)$ be the object that points to agent $i$ in the cycle $\mathcal{C}$, e.g., if $\mathcal{C}=(o\to i\to o'\to j\to o)$, then $e_i(R)=o$ and $e_j(R)=o'$.

Next, we inductively show that for each agent $i\in N$, $e_i(R)$ is in his trading opportunity set.

\begin{lemma}\label{Le_th2}
For each $R\in\mathcal{R}^N$, and each agent $i\in N$, $e_i(P) \in T_{i}(P)$.
\end{lemma}
\begin{proof}[Proof of Lemma~\ref{Le_th2}]
Let $R\in\mathcal{R}^N$, suppose that $TTC(R)$ terminates after Step~$K$.
Next we prove the lemma by the induction on $\tau=1,\ldots,K$.

\noindent\textbf{\textit{Induction basis.}} 
For each center $c\in C$, consider the agent in $c$ who has the highest priority according to $\succ_c$. 
By the definition of $\mathcal{T}$, we see that his trading opportunity set equals the set of objects in his center, $O_c$.
By the definition of TTC, for each $\mathcal{C}\in \mathcal{C}_1(R)$, and each $i\in S_{\mathcal{C}}$, we know that (1) $i$ has the highest priority according to $\succ_{c(i)}$; and (2) $e_i(R)\in O_{c(i)}$.
Thus, for each $i\in S_{\mathcal{C}}$, $e_i(R) \in T_{i}(R)=O_{c(i)}$.

\noindent\textbf{\textit{Induction hypothesis.}}
Let $L\in \{2,\ldots,K\}$, suppose that for each $\tau <L$, each $\mathcal{C}\in \mathcal{C}_\tau(R)$, and each $i\in S_{\mathcal{C}}$, we have $e_i(R) \in T_{i}(R)$.

Before we go to the induction step, observe that for $\tau <L$, if an agent $i$ and $e_i(R)$ is removed at step~$\tau$ under $TTC(R)$, then they are also removed at step~$\tau$ under $\mathcal{T}$, according to $R$.
This is because under both algorithms, at step~$\tau$, $i$ points to his most-preferred remaining object according to $R_i$ and $e_i(R)$ points to $i$ at step~$\tau$. Thus, $i$ and $e_i(R)$ are involved in a removed cycle at step~$\tau$ under both algorithms as well.

\noindent\textbf{\textit{Induction step.}}
Let $\mathcal{C}\in \mathcal{C}_L(R)$ and $i\in S_{\mathcal{C}}$.
Recall that at the beginning of step~$L$ under $TTC(R)$, the remaining  objects are $O^L$. Hence, $T_i(R)\subseteq O^L\cap O_{c(i)}$.
By the definition of TTC, we know that at step~$L$, (1) $e_i(R)$ belongs to $c(i$) and is remaining, i.e., $e_i(R)\in O^L\cap O_{c(i)}$; and
(2) agent $i$ has the highest priority among the remaining agents in center $c(i)$ according to $\succ_{c(i)}$. 
Thus, by the definition of $\mathcal{T}$ and the observation above, we see that all remaining objects in $O^L\cap O_{c(i)}$ point to agent $i$ at step~$L$ of $\mathcal{T}$, according to $R$. As a result, $e_i(R)$ points to $i$ under $\mathcal{T}$. Since $\mathcal{C}$ and $i$ are arbitrary chosen, this argument holds for each agent involved in $\mathcal{C}$.
Thus, we conclude that for each $\mathcal{C}\in \mathcal{C}_L(R)$, and each involved agent $i\in  S_{\mathcal{C}}$, $e_i(R)$ to points to $i$ under $\mathcal{T}$.
Again, by the observation above, given $R$, we see that all agents involved in $ \mathcal{C}$ also form a cycle at step~$L$ under $\mathcal{T}$. Therefore, for each $i\in S_{\mathcal{C}}$, $e_i(R) \in T_{i}(R)$. 
\end{proof}  

Together with Lemma~\ref{Le_th2}, \emph{procedural fairness} of $f$ implies that 
\begin{equation}
 \text{for each } R\in\mathcal{R}^N,\text{ and each }i\in N,
f_i(R) \mathbin{R_i} e_i(R). \label{equ:IR}
\end{equation}
Otherwise $i$ can be better off by consuming $e_i(R)\in T_i(R)$.
Given this fact, for each $R\in\mathcal{R}^N$ and each $i\in N$, we 
say an object is \textit{acceptable} at $R_i$ if $o\mathbin{R_i} e_i(R)$.

Now, we adapt \citet{Sethuraman2016}'s approach to complete the proof of Theorem~\ref{thm:pf}.
For each $R\in\mathcal{R}^N$, the \textit{size} of $R$, $s(R)$, is defined as the total number of acceptable objects across all agents at $R$, i.e., $s(R)=\sum_{i\in N}|\{o\in O\mid o\mathbin{R_i}e_i(R)\}|$.

Suppose that $f\neq TTC$. A profile $R\in\mathcal{R}^N$ is \textit{discordant} if $f(R)\neq TTC(R)$. By our assumption we know that there exists at least one discordant profile.
From now, we define $R$ as a profile with the minimal size, i.e., for each $\tilde{R}\in\mathcal{R}^N$, if $\tilde{R}$ is discordant, then $s(R)\leq s(\tilde{R})$.

Let $I_f(R)=\{i\in N\mid f_i(R)\mathbin{P_i}TTC_i(R)\}$ be the set of agents who prefer $f$ to $TTC$ at $R$. Similarly, $I_{TTC}(R)$ is defined analogously.

\begin{claim}
For each $i\in I_f(R)\cup I_{TTC}(R)$, agent $i$ has exactly two acceptable objects at $R_i$.\label{claim:size}
\end{claim}  
\begin{proof}[Proof of Claim~\ref{claim:size}]
By contradiction, it is without loss of generality to assume that there exists an agent $i\in I_f(R)$ who has more than two acceptable objects at $R_i$. Let $R'_i:f_i(R),e_i(R),\ldots$ and $R'\equiv (R'_i,R_{-i})$.

By \emph{strategy-proofness} of $f$, $f_i(R)=f_i(R')$; otherwise agent $i$ has an incentive to misreport $R_i$ at $f(R')$. Since $i\in I_f(R)$, $f_i(R)\mathbin{P_i}TTC_i(R)$. Thus, by \emph{strategy-proofness} of $TTC$, $TTC_i(R')\neq f_i(R)$; otherwise $i$ has an incentive to misreport $R'_i$ at $TTC(R)$. Overall, $f_i(R')=f_i(R)\neq TTC_i(R')$.
Thus, $R'$ is discordant such that $s(R') <s(R)$, contradicting the choice of $R$.
\end{proof}  

By \citet{pycia2017}, we know that TTC is \emph{Pareto efficient}, and hence $I_{TTC}(R)$ is non-empty, otherwise it means $f$ Pareto dominates TTC at $R$.

Let $i\in I_{TTC}(R)$, and let $\mathcal{C}\in \mathcal{C}(R)$ be the cycle involving agent $i$. In other words, $i\in S_{\mathcal{C}}$.

By Claim~\ref{claim:size}, we know that $R_i$ is such that $R_i:TTC_i(R),e_i(R),\ldots$

Since $i\in I_{TTC}(R)$, $\mathcal{C}$ is not executed at $f(R)$. 
Thus, $f_i(R)\neq TTC_i(R)$. By (\ref{equ:IR}), $f_i(R)\in \{e_i(R),TTC_i(R)\}$.
Therefore, $f_i(R)=e_i(R)$. This implies that the agent who points to $e_i(R)$ in $\mathcal{C}$, say agent $j$, does not receive $e_i(R)$ at $f(R)$. By applying similar arguments to agent $j$, we see that $f_j(R)=e_j(R)$. 
By repeating this argument for all agents in $S_{\mathcal{C}}$, we conclude that for each agent $k\in S_{\mathcal{C}}$, $f_k(R)=e_k(R)$. Consequently, $
S_{\mathcal{C}}\subseteq I_{TTC}(R)$.

Without of loss of generality, let $S_{\mathcal{C}}=\{i_1,\ldots,i_L\}$.
If $L=1$ then $TTC_{i_1}(R)=e_{i_1}(R)$. Since $f_{i_1}(R)\mathbin{R_{i_1}} e_{i_1}(R)$ we must have $f_{i_1}(R)=e_{i_1}(R)=TTC_{i_1}(R)$, contradicts $S_{\mathcal{C}}\subseteq I_{TTC}(R)$.
Thus, $L>1$.

By the definition of TTC, we know that all agents in $S_{\mathcal{C}}$, are form the different centers.
Thus, there is a set that consists of $L$ different centers $\{c({i_1}),\ldots,c({i_L})\}$. For each $i_\ell\in  S_{\mathcal{C}}$, $TTC_{i_\ell}(R)=e_{i_{\ell+1}}(R)\mathbin{P_{i_\ell}}e_{i_\ell}(R)=f_{i_\ell}(R)$ (mod $L$).
Also, Lemma~\ref{Le_th2} implies that $e_{i_\ell}(R)\in T_{i_\ell}(R)$.
Therefore, $f(R)$ is \emph{unfairly produced}, as agents $\{i_1,\ldots,i_L\}$ can be better off by exchanging from their trading opportunity sets, contradicts \emph{procedural fairness} of $f$.
This completes the proof of Theorem~\ref{thm:pf}.

\section{Relation to queuewise rationality}
\label{sec:qr}
\citet{ekici2024} characterizes TTC through \emph{strategy-proofness}, \emph{pair efficiency}, and \emph{queuewise rationality}. 
Here, we further discuss the relationship between his result and ours.

By Lemma~\ref{lemma:pfqr}, we know that \emph{procedural fairness} is mathematically stronger than \emph{queuewise rationality}. However, their economics interpretations are distinct: \emph{procedural fairness} is a fairness property based on agents' trading opportunities, \emph{queuewise rationality} is a voluntary participation condition. 

Next, we show that our other fairness properties, \emph{internal fairness} and \emph{external fairness} are logically unrelated to \emph{queuewise rationality}.
To illustrate this, we provide three examples below, each satisfying two of the properties but violating the third.
\medskip

For simplicity, we focus on the case in which $C=\{c,c'\}$, $N_c=\{1\}$, $O_c=\{o_1\}$, $N_{c'}=\{2,3\}$, $O_{c'}=\{o_2,o_3\}$, and $\succ_{c'}:2,3$.

\begin{example}[Serial-QR mechanism]\ \\
\label{example:sqr}
We adapt so-called Multiple-Serial-IR mechanisms introduced by \citet{biro2021}.

\noindent \rule{1\columnwidth}{1pt}
\paragraph{Serial-QR Algorithm:} 
\begin{description}
\item[Step~$0$]: Let $\delta=(1,2,3)$ be an order of $N$. Let $R\in\mathcal{R}^N$. Let $Y(0)\subseteq X$ be the set of allocations that satisfy \emph{queuewise rationality} at $R$.
\item[Step~$k\geq 1$:] 
Let $Y_k$ be the set of agent $k$'s s allotments that are compatible with some allocation in $Y(k-1)$, i.e., $Y_k$ consists of all $y_k$ for which there exists an allocation $x\in Y(k-1)$ such that $x_k=y_k$. 
Let $y^*_k$ be agent $k$’s most preferred allotment in $Y_k$, i.e., for each $y_k\in Y_k$, $y^*_k\mathbin{R_k} y_k$.
Let $Y(k)\subseteq Y(k-1)$ be the set of allocations in $Y(k-1)$ that are compatible with $y^*_k$, i.e., $Y(k)$ consists of all $x\in Y(k-1)$ with $x_k=y^*_k$. 

\item[Termination:] The algorithm terminates after step~3.
\end{description}
\noindent \rule{1\columnwidth}{1pt}

Let $\Delta$ be such that for each $R\in \mathcal{R}^N$, $\Delta(R)$ is determined by the algorithm above.
By the definition of $\Delta$, it satisfies \emph{queuewise rationality} and \emph{internal fairness}.
To see it violates \emph{external fairness}, consider $R$ as follows.
$$R_1:o_2,o_1,o_3;$$
$$R_2:o_3,o_2,o_1;$$
$$R_3:o_2,o_1,o_3.$$

Note that by \emph{queuewise rationality}, for each $R'_{-2}$, $\Delta_2(R_2,R'_{-2})=o_3$ and hence agent $2$ does not depend on agent $1$. $\Delta(R)=(o_2,o_3,o_1)$ and hence agent $3$ envies agent $1$. Since agents $1$ and $3$ are from different centers, $o_2$ belongs to agent $3$'s center, and agent $2$ with higher priority than $3$ does not depend on agent $1$, we see that agent $3$'s envy is \emph{externally justified}.
     \hfill $\diamond$  
\end{example}

\begin{example}[Two steps: SD then TTC]\ \\
Given $R$, for each center $c$, first run serial dictatorship over $O_c$ with respect to $\succ_c$:
the highest-priority agent in $N_c$ picks up his most preferred object in $O_c$,
the second-highest priority agent picks his most preferred object among the remaining objects in $O_c$, and so on.
Denote the resulting allocation by $e(R)$.
Then, let $e(R)$ be an ``artificial'' endowment allocation and run the TTC algorithm with respect to $e(R)$, i.e., let each $e_i(R)$ points to $i$. 
This two-step mechanism satisfies \emph{queuewise rationality} and \emph{external fairness}, but violates \emph{internal fairness}. \label{example:TTC}
     \hfill $\diamond$  
\end{example}

\begin{example}[A variant of serial dictatorships]\ \\
\label{example:vsd}
Let $SD$ be the serial dictatorship algorithm with $\delta:1,2,3$. That is, agent $1$ picks up his most preferred object in $O$, agent $2$ picks up his 
his most preferred object among the remaining objects, and the last object is given to agent $3$.

Let $f$ be such that if agents $1$ and $2$ both have the same most preferred object, and that object is in $O_{c'}$, then $f(R)=TTC(R)$; otherwise $f(R)=SD(R)$.
$f$ satisfies \emph{internal fairness} and \emph{external fairness}, but violates \emph{queuewise rationality}.
\hfill $\diamond$    
\end{example}

\subsection{An alternative characterization}
\label{sec:another}
Here,  adapted from \citet{chen2025welfare}'s \emph{vacant object lower bounds},
we weaken \emph{queuewise rationality} to the following property, the \emph{center lower bound}, which imposes that each agent receives an object that is at least as desirable as the least desirable object owned by the center to which he belong.

For each $R\in \mathcal{R}^N$, each $c\in C$, and each $i\in N_C$, let $\underline{o}_i(R)\in O_c$ be the worst object in $O_c$ at $R_i$. That is, for each $o\in O_c$, $o\mathbin{R_i} \underline{o}_i(R)$.
\smallskip

\noindent The \textbf{center lower bound}: For each $R\in \mathcal{R}^N$ and each $i\in N$, $f_i(R)\mathbin{R_i} \underline{o}_i(R)$.\smallskip

It is easy to see that the \emph{center lower bound} 
is weaker than \emph{queuewise rationality}.
Also, from examples above, we see that the \emph{center lower bound} is logically independent with \emph{internal fairness} and \emph{external fairness}.

Next, from Theorem~\ref{thm}, we characterize TTC by replacing \emph{external fairness} with the \emph{center lower bound}. 

\begin{theorem}
    A mechanism $f$ satisfies \textbf{strategy-proofness}, \textbf{pair efficiency},
    \textbf{internal fairness}, and the \textbf{center lower bound}, if and only if $f=TTC$.  \label{thm:ir}
\end{theorem}

The proof is similar to Theorem~\ref{thm}, and hence we omit it.
Also note that as in Theorem~\ref{thm}, one can weaken \emph{internal fairness} to \emph{weak internal fairness} in Theorem~\ref{thm:ir}.

\section{Relation to stability}
\label{sec:stable}
\subsection{Core-stability}
In this subsection, by adopting the concept of strict core~/~core-stability, we further explore TTC from the perspective of cooperative game theory. 
In Shapley-Scarf housing markets, a set of agents $S\subseteq N$ can block an allocation $x$ by reassigning their own endowments among themselves to make all members of $S$ weakly better off and at least one member strictly better off. 
The definition of blocking can be decomposed into two conditions: one is the weak dominance condition, which says that for a set of agents $S$ to form blocking, all members in $S$ are weakly better off and at least one member is strictly better off; and another one is the self-autarky condition, which states that for $S$ to form a blocking, each member must bring his own endowment.

Still, the first condition can be directly applied to our context. However, the second condition needs to be modified because in our model, an agent may not own a single object, but rather a set of objects associated with the center to which he belongs. 
Thus, it is natural to require that what agents can bring to form a blocking based on the underlying priority structure.
Aligning with our idea of \emph{procedural fairness}, here we adjust the self-autarky condition by requiring that each member brings one object from his trading opportunity set.

Let $R\in \mathcal{R}^N$.
At $R$, a set of agents $S\subseteq N$ \textit{properly blocks} an allocation $x$ via $y$ if
\begin{itemize}
    \item[(1)] for each $i\in S$, $y_i\mathbin{R_i}x_i$ and for some $j\in S$, $y_j\mathbin{P_i}x_j$; 
    \item[(2)] there exists a one to one function $\omega:S\to O$ such that (2-a) for each $i\in S$, $\omega(i) \in T_i(R)$; and (2-b) $\{y_i\}_{i\in S}=\{\omega(i)\}_{i\in S}$.
\end{itemize}

Two points are worth noting about our self-autarky condition.
First, in the special case where each center contains only one agent, it reduces to the classical self-autarky condition used in the definition of the core.
Second, recall that an agent's trading opportunity set depends on the preferences of agents with higher priority. This restriction implicitly implies that whether an agent can (or cannot) bring an object to form a blocking depends on the views of higher-priority agents.\footnote{Aligning with this, \citet{ekici2024} modifies the self-autarky condition in a different way: he requires that for each $i\in S$, all agents who are from the same center and have higher priorities than $i$ must also be in $S$.}

An allocation is \textit{properly core-stable} at $R$ if it is not \emph{properly blocked} by any set of agents.
Such an allocation is also referred to as an \textit{proper core allocation}.
Note that any \emph{proper core allocation} is \emph{fairly produced}, but not vice versa.
The \textit{proper core}, $PC(R)$ is the set that contains all \emph{proper core allocations} at $R$.
\medskip

We next show that the \emph{proper core} is single-valued: Particularly, it is induced by TTC.

\begin{theorem}
    For each $R\in\mathcal{R}^N$, $PC(R)=\{TTC(R)\}$. \label{thm:PC}
\end{theorem}

\subsubsection{Proof of Theorem~\ref{thm:PC}}

First, we provide one useful result about agents' trading opportunity sets.
For each $R\in\mathcal{R}^N$, let $I_\tau(R)\subseteq N$ and $O_\tau(R)\subseteq O$ be the set of agents and the set of objects that are removed at the end of step~$\tau$ under $TTC(R)$, respectively.

\begin{claim}
For each $t$, and each $j \in I_t$, $T_j(R) \cap \bigcup_{\tau=1}^{t-1} O_\tau(R) = \emptyset$. \label{c:smaller}
\end{claim}

The proof immediately follows the definition of $\mathcal{T}$ and hence we omit it.
The rest of proof is similar to the proof in \citet{roth1977}, who demonstrate that the TTC always induces the unique core allocation in Shapley-Scarf housing markets.

We first show that TTC induces a proper core allocation.
Suppose that TTC does not induce a proper core allocation. Then, there are $R\in\mathcal{R}^N$ and $S\subseteq N$ such that $S$ properly blocks $TTC(R)$ via an allocation $y\neq TTC(R)$. Let $x\equiv TTC(R)$.
Let $\mathcal{C}\in\mathcal{C}_1(R)$ be a trading cycle formed at the first step under $TTC(R)$, and recall that $S_{\mathcal{C}}\subseteq N$ is the set of agents who are involved in $\mathcal{C}$. There are two cases.
First, if an agent $i\in S_{\mathcal{C}}$ is also in $S$, then we have $y_i=x_i$ as $x_i$ is his most preferred object.
Then, by our self-autarky condition (2), there is an agent $\widehat{i}\in S$ such that $x_i\in T_{\widehat{i}}(R)$. It implies that $\widehat{i}\in 
S_{\mathcal{C}}$; otherwise $i$ cannot receive $x_i$ via TTC.
Generally, we find that for each object involved in $\mathcal{C}$ there is a unique agent whose trading opportunity set contains that object, namely, the agent in $S_\mathcal{C}$ to whom the object points. Together with our self-autarky condition (2), it implies that agent is also in $S$.
Consequently, we conclude that $S_{\mathcal{C}}\subseteq S$, and all agents in $S_\mathcal{C}$ receive their TTC allocation at $y$ as well.
This implies that no agent in $S$ can be strictly better off by receiving an object from $\{x_i\}_{i\in S_{\mathcal{C}}}$.
Second, if no agent in $S_{\mathcal{C}}$ is in $S$, then by Claim~\ref{c:smaller}, we see that for each agent in $S$, his trading opportunity set does not contain any object from $\{x_i\}_{i\in S_{\mathcal{C}}}$. 
In any case, we conclude that no agent in $S$ can be strictly better off by receiving an object from an object that is removed from the first step of TTC.
Applying the argument to the second step and so on, we obtain that there is no agent $j\in S$ such that $y_j\mathbin{P_j}x_j$, leading to a contradiction. 

The converse can be proved by adopting Lemma~1 in \citet{roth1977}. Let $R\in\mathcal{R}^N$, $x\equiv TTC(R)$, and $\mathcal{C}_1,\ldots,\mathcal{C}_K$ be corresponding trading cycles. Without loss of generality, assume that for any $t,\tau=1,\ldots,K$, $t<\tau$ means that $\mathcal{C}_t$ is executed no later than $\mathcal{C}_\tau$. We also have corresponding coalitions $S_{\mathcal{C}_1},\ldots,S_{\mathcal{C}_K}$.
Now, consider an arbitrary allocation $y\in X$ with $y\neq x$. 
Let $K^*\in\{1\ldots,K\}$ be the first integer where there exists an agent $j\in S_{\mathcal{C}_{K^*}}$ such that $x_j\neq y_j$ Then, by construction, coalition $S=\cup_{\kappa=1}^{K^*} S_{\mathcal{C}_\kappa}$ properly blocks $y$ via $x$. To see it, we check our two conditions one by one.
By the selection of $K^*$, for each $i\in S\setminus S_{\mathcal{C}_{K^*}}$, $x_i=y_i$; by the definition of TTC, for each $i\in S_{\mathcal{C}_{K^*}}$ and each $o\in O\setminus\{x_i\}_{i \in S\setminus S_{\mathcal{C}_{K^*}} }$, $x_i\mathbin{R_i}o$. Thus, for each $i\in S_{C_{K^*}} $, 
$x_i\mathbin{R_i}y_i$, and hence $x_j\mathbin{P_j}y_j$, as $x_j\neq y_j\in O\setminus\{x_i\}_{i \in S\setminus S_{\mathcal{C}_{K^*}}}$. That is, condition (1) is satisfied. By the definition of TTC, we see $\{x_i\}_{i\in S}$ satisfies condition (2).


\subsection{Pairwise stability}\label{sec:pairwisestable}
For priority-based allocation problems, e.g., school choice problems, a standard stability notion is \emph{pairwise stability} \citep{GaleShapleyAMM1962}: An allocation $x$ is \textit{pairwisely stable} at $R$ if there is no agent-object pair $(i,o)$ such that $o\mathbin{P_i}x_i$ and $i\succ_o j$, where $x_j=o$.
Clearly, this notion is not suitable here, as in our model, objects have no priorities over all agents. Even if we adopt each center $c$'s priority for objects in $O_c$, it would still be problematic because $\succ_c$ is incomplete.  
One way to address this issue is to expand the priorities to be complete. A natural expansion is as follows: For each $c\in C$,
each $o\in O_c$, each $i\in N_c$, and each $j\not\in N_c$, $i$ has higher expanded priority at $o$ than $j$.\footnote{\citet{klaus2019} propose an extension of TTC based on such expanded priorities, and \citet{combe2023reallocation} shows that this TTC extension is ``minimal envy'' among the set of \emph{individually rational}, \emph{strategy-proof}, and \emph{Pareto efficient} mechanisms.}
However, this assumption will lead to an incompatibility with \emph{pair efficiency}. To see it, let us revisit Example~\ref{example1}. Suppose that at $a$ and $b$, agent $1$ has higher priority than $3$, as $a,b\in O_{c},1\in N_c$, and $3\not \in N_c$. Similarly, 
$3$ has higher priority than $1$ and $2$ at $d$.
Suppose that $2$ receives $d$: in this case, if $3$ receives $b$, then $3$ would block this allocation as $3$ has higher priority than $2$ at $d$; however, if $3$ receives $a$ then it means $1$ receives $b$. So, $1$ would block this allocation as $1$ has higher priority than $3$ at $a$. Thus, $2$ does not receive $d$. Since $2$ also has priority than $1$ at $a$, we conclude that it must be the case that $2$ receives $a$. Subsequently, $1$ receives $b$ and $3$ receives $d$, i.e., the allocation $y$. However, this allocation is not \emph{pair efficient}: agents $2$ and $3$ can be better off by switching their allotments.

Therefore, finding alternative ways to expand priorities to be complete and defining criteria to select ``fair'' expansions constitutes an interesting open problem, though one that is
beyond the scope of this paper.

\section{Conclusion}
\label{sec:con}
In this paper, we investigate \citet{ekici2024}'s multi-center allocation problems. By proposing different fairness properties, we provide several TTC characterizations.
Our analysis sheds light on \citet{ekici2024} to show that TTC performs surprisingly well according to different objectives, e.g., efficiency, fairness, and incentives.
We finish this paper with some final remarks related to issues left aside during the presentation of the main results.

\subsection{Length constraints}
The primary application target of this paper is cross-program cooperation in organ exchange programs, e.g., kidney and liver exchanges.
One might wonder if TTC is not applicable for kidney exchanges, as it allows trading cycles without length constraints (of course, in our model, each cycle contains at most $|C|$ agents).\footnote{If $|C|$ is sufficiently small, we may not even meet length constraints. Note that in India, cyclic exchanges involving up to ten pairs have been successfully conducted \citep{kute2021non}.} 
In practice, cycle length constraints are imposed due to the scarcity of medical resources: a two-way exchange requires the simultaneous availability of four operating rooms and associated personnel, while a three-way exchange requires six, and so on.
However, in our model, the presence of multiple centers allows for 
integrating medical resources from different centers, making it feasible to conduct exchanges with longer cycles.\footnote{There are also studies that examine cooperation across different kidney exchange programs without considering length constraints; see, e.g., \citet{benedek2024computing}.}

\subsection{Balancedness}
Another restriction in our model is balancedness: For each center $c$, the inflow of agents coming from other centers to receive objects in $O_c$ is equal to the outflow of agents in $N_c$ who leave $c$ to receive objects in other centers.
This restriction aligns with real-world scenarios in kidney exchanges: As in the classical kidney exchange model, assuming $|N_c|=|O_C|$ is based on the observation that, in most cases, each patient comes with a living kidney donor, often a family member.\footnote{For organ exchange environments with imbalance and capacity constraints, \citet{ekici2024} proposes a solution by introducing fictitious agents and objects; the same construction can be applied here as well.} 
While this restriction seem to limit the applicability of this paper to other use cases, we argue this point by the following reason: Balancedness is an inviolable constraint in many practical allocation problems, such as shift reallocation and exchange programs \citep{dur2019two,yu2020market,manjunath2021,kamada2023}.
As pointed out by \citet{dur2019two}, in such settings, every position and every worker must be matched, and hence the balancedness restriction is imposed.
Thus, balancedness does not reduce the applicability of this paper.
Of course, considering extensions involving imbalance or varying object capacities is an interesting open problem.
However, such issues are beyond the scope of this paper, and we leave them for future research.


\bibliographystyle{econ-econometrica}
\phantomsection\addcontentsline{toc}{section}{\refname}\bibliography{BibFile}

\appendix

\section{Proofs}
\label{appendix:proof}

\subsection{Proof of Lemma~\ref{lemma:pfqr}}
Let $f$ be \emph{procedurally fair}.
Let $R\in\mathcal{R}^N$ and $x\equiv f(R)$. Let $i\in N$, $c\equiv c(i)$. Suppose that $i$ has the $k$-highest priority at $c$, i.e., $r_i=k$. We see that at most $k-1$ objects in $O_c$ are not in $T_i(R)$.
Consider the worst case that all the $k-1$ number of objects are more preferred than $o_i(R)$. In this case, $o_i(R)\in T_i(R)$. Therefore, $x_i\mathbin{R_i} o_i(R)$ holds, because at least agent $i$ can directly consume $o_i(R)$ from $T_i(R)$.

Next, we use an example to show that a \emph{queuewise rational} allocation may not be \emph{fairly produced}.

Let $C=\{c_1,c_2\}$, $N_{c_1}=\{1,2\},N_{c_2}=\{3\}$, $O_{c_1}=\{o_1,o_2\},O_{c_2}=\{o_3\}$, and $\succ_{c_1}:1,2$. 

Let $R$ be such that 
$$R_1:o_3,o_1,o_2;$$
$$R_2:o_3,o_2,o_1;$$
$$R_3:o_1,o_2,o_3.$$

Let $y\equiv (o_1,o_2,o_3)$. We see that $y$ is \emph{queuewise rational} at $R$. Note that $T_1(y)=\{o_1,o_2\}$ and $T_3(y)=\{o_3\}$. We find that agents $1$ and $3$ can be better off by trading $(o_1,o_3) \in T_1(y) \times T_3(y)$. Thus, $x$ is \emph{unfairly produced} at $R$.

\subsection{Proof of Proposition~\ref{proposition:sufficicent}}
Recall that TTC is \emph{pair efficient} and \emph{strategy-proof} as it is a sub-class of trading cycle mechanisms \citep{pycia2017}.

\emph{Internal fairness}.
Suppose that TTC violates \emph{internal fairness}. Then, there is a profile $R$, a center $c\in C$, and a pair of agents $i,j\in N_C$ such that $i\succ_c j$, and $i$ \emph{envies} $j$ at $TTC(R)$.
If $TTC_j(R) \in O_c$, then by the definition of TTC, we know that $i$ can receive $TTC_j(R)$ by pointing to $TTC_j(R)$ as $i\succ_c j$.
If $TTC_j(R)\not \in O_c$. Let $R'_i:TTC_j(R),\ldots$ We see that $TTC_i(R'_i,R_{-i})=TTC_j(R)$ by the definition of TTC, as $i$ moves earlier than $j$. This implies that TTC is not \emph{strategy-proof}, a contradiction.

\emph{External fairness}.
By the definition of TTC, we know that each agent in a cycle depends on each of the others in that cycle (or agents in the previous cycles).
Also we know that, if $i$ \emph{envies} $j$, then $j$ receives his allotment, denoted by $o$, at an earlier step than $i$.
Let $\mathcal{C}$ be the cycle that includes $j$. So, $i\in N_{c(i)}$ and $j\not\in N_{c(i)}$, and recall that $o\in O_{c(i)}$.
In $\mathcal{C}$, $o$ is pointing to an agent $k\in N_{c(i)}$ who is higher ranked than $i$ (i.e., $k\succ_{c(i)} i$) and who depends on $j$. Therefore, this envy is not \emph{externally justified}.

\emph{Procedural fairness}.
Finally, we inductively show that TTC is \emph{procedurally fair}.
Let $R\in\mathcal{R}^N$. Suppose that $TTC(R)$ stops at step~$K$.

\noindent\textbf{\textit{Induction basis.}}
Let $\mathcal{C}\in \mathcal{C}_1(R)$ and $i\in S_\mathcal{C}\subseteq I_1$. By the definition of TTC, we know that $x_i$ is $i$'s most preferred object.  
Thus, we conclude no agent in $I_1$ can be better off either by directly consuming an object from his trading opportunity set or by trading with other agents.


\noindent\textbf{\textit{Induction hypothesis.}} Let $\kappa\leq K$.
Suppose that for $\tau=2,\ldots,\kappa-1$,
no agent in $I_\tau$ can be better off either by directly consuming an object from his trading opportunity set or by trading with other agents.

\noindent\textbf{\textit{Induction step.}}
We show that no agent in $I_\kappa$ can be better off either by directly consuming an object from his trading opportunity set or by trading with other agents. Let $O^\kappa$ be the set of remaining objects at the beginning of step~$\kappa$ under $TTC(R)$. Let $i\in I_\kappa$.

Case (a). Direct consumption.
By the definition of TTC, $TTC_i(R)$ is $i$'s most preferred object in $O^\kappa$. By Claim~\ref{c:smaller}, $T_i(R)\subseteq {O}^\kappa $. 
Thus, $i$ cannot be better off by directly consuming an object from his trading opportunity set $T_i(R)$.

Case (b). Exchange.
By contradiction, suppose that $i$ can be better off by trading with other agents.
Thus, there is a list of centers $c_1, \ldots, c_L \in C$, a list of agents $(i_1, \ldots, i_L) \in \prod_{\ell=1}^L N_{c_\ell}$, and a list of objects $(o_1, \ldots, o_L) \in \prod_{\ell=1}^L T_{i_\ell}(R)$ such that for each $\ell \in \{i_1,\ldots,i_L\}$, $o_{\ell+1}\mathbin{P_\ell}TTC_\ell(R)$ (mod $L$). Without of loss of generality, let $i_1=i$. By Case (a) we know that $L>1$. Thus, $i_2\neq i$.
Now consider $i_2$. By the induction hypothesis, no agent in $\cup_{\tau=1}^{\kappa-1}I_\tau$ can be better off by trading with $i$. 
Thus, $i_2$ is remaining at step~$\kappa$, i.e., $i_2 \not \in \cup_{\tau=1}^{\kappa-1}I_\tau$. 
Then, by Claim~\ref{c:smaller}, we have $T_{i_2}(R)\subseteq {O}^\kappa $.
Also, recall that for each $o\in O^L$, $TTC_i(R)\mathbin{R_i} o$. 
Therefore, $i$ cannot be better off by trading with $i_2$.

\section{Independence examples}
\label{appendix:examples}
We establish the independence of properties through the following examples.
We label examples by the property that is not satisfied.

\subsection*{Theorem~\ref{thm}}
\begin{example}[Staregy-proofness]\ \\
Let $C=\{c_1,c_2,c_3\}$, $N_{c_1}=\{1\}$, $O_{c_1}=\{o_1\}$, $N_{c_2}=\{2\}$, $O_{c_2}=\{o_2\}$, $N_{c_3}=\{3,4\}$, $O_{c_3}=\{o_3,o_4\}$, and $\succ_{c_3}:3,4$.

Consider a ``circular preference profile'' as the one below.
$$\widetilde{R}_1:o_2,o_3,o_1,o_4;$$
$$\widetilde{R}_2:o_3,o_1,o_2,o_4;$$
$$\widetilde{R}_3:o_1,o_2,o_3,o_4.$$
$$\widetilde{R}_4:\ldots$$

Note that $x\equiv (o_3,o_1,o_2,o_4)$ is \emph{pair efficient} at $\widetilde{R}$. However, $TTC(\widetilde{R})$ outputs the allocation $(o_2,o_3,o_1,o_4)$, which makes agents $\{1,2,3\}$ strictly better off.

Let $f$ be such that $f(\widetilde{R})= x$ and, for each $R\neq \widetilde{R}$, let $f(R)=TTC(R)$. $f$ is \emph{pair efficient},   \emph{internally fair}, \emph{externally fair}, but is not \emph{strategy-proof}.  \label{example:sp}
 \hfill $\diamond$    
\end{example}

\begin{example}[Pair efficiency]\ \\
Let $C=\{c_1,c_2\}$, $N_{c_1}=\{1\}$, $O_{c_1}=\{o_1\}$, $N_{c_2}=\{2,3\}$, $O_{c_2}=\{o_2,o_3\}$, and $\succ_{c_2}:,2,3$.

Let $f$ be such that for each $R\in\mathcal{R}^N$,
$f_1(R)=o_1$ and $f_2(R)\mathbin{R_2}f_3(R)$.  
Essentially, it means that agent $1$ always receives his own, and we run a serial dictatorship algorithm with $2,3$.
We see that $f$ violates \emph{pair efficiency} and satisfies other three properties.
 \hfill $\diamond$    \label{example:DA}
\end{example}

\begin{example}[Internal fairness]\ \\
Let $|C|=1$. Thus, \emph{external fairness} is trivially satisfied.
Run the serial dictatorship algorithm that does not respect the center's priority $\succ_c$. It is easy to see that it is \emph{strategy-proof} and \emph{pair efficient} but is not \emph{internally fair}.
 \hfill $\diamond$ \label{example:RSD}
\end{example}

\begin{example}[External fairness]\ \\
Let $C=\{c_1,c_2\}$, $N_{c_1}=\{1\}$, $O_{c_1}=\{o_1\}$, $N_{c_2}=\{2,3\}$, $O_{c_2}=\{o_2,o_3\}$, and $\succ_{c_2}:,2,3$.

Let SD be the serial dictatorship with $\pi:2,1,3$.
It is easy to see that SD is \emph{strategy-proof}, \emph{pair efficient}, and \emph{internally fair} but is not \emph{externally fair}. 
\hfill $\diamond$    
\end{example}

\subsection*{Theorem~\ref{thm:pf} and Procedural fairness}
We first show that \emph{strategy-proofness} and \emph{procedural fairness} are independent.

\noindent \textbf{Staregy-proofness}: 
Consider a modification of Example~\ref{example:sp} as follows.

Let $C=\{c_1,c_2,c_3\}$, $N_{c_1}=\{1\}$, $O_{c_1}=\{o_1\}$, $N_{c_2}=\{2\}$, $O_{c_2}=\{o_2\}$, $N_{c_3}=\{3,4\}$, $O_{c_3}=\{o_3,o_4\}$, and $\succ_{c_3}:3,4$.

Let $\widetilde{R}$ be such that

$$\widetilde{R}_1:o_2,o_3,o_1,o_4;$$
$$\widetilde{R}_2:o_1,o_3,o_2,o_4;$$
$$\widetilde{R}_3:o_1,o_2,o_3,o_4.$$
$$\widetilde{R}_4:\ldots$$

Let $f$ be such that $f(\widetilde{R})= (o_2,o_3,o_1,o_4)$ and, for each $R\neq \widetilde{R}$, $f(R)=TTC(R)$. $f$ is \emph{procedurally fair} but is not \emph{strategy-proof}.
 \hfill $\diamond$  
\medskip

\noindent \textbf{Procedural fairness}: A constant mechanism that always selects the same allocation is \emph{strategy-proof} but is not \emph{procedurally fair}. \hfill $\diamond$  
\medskip

Next, we show that \emph{procedural fairness} is independent with other mentioned properties. By Theorem~\ref{thm:pf} and independence examples in Theorem~\ref{thm}, one can find that \emph{pair efficiency}~/~\emph{internal fairness}~/~
\emph{external fairness} cannot imply \emph{procedural fairness} solely. Thus, we only show that \emph{procedural fairness} does not imply any of the other properties.
\medskip


\noindent \textbf{Pair efficiency}: 
Let $C=\{c_1,c_2,c_3\}$, $N_{c_1}=\{1,2\},N_{c_2}=\{3\},N_{c_3}=\{4\}$, $O_{c_1}=\{o_1,o_2\},O_{c_2}=\{o_3\},O_{c_3}=\{o_4\}$, and $\succ_{c_1}:1,2$. 

Let $R$ be such that 
$$R_1:o_3,o_4,\ldots;$$
$$R_2:o_4,o_3,\ldots;$$
$$R_3:o_1,o_2,\ldots;$$
$$R_4:o_2,o_1,\ldots.$$


Consider $x\equiv (o_3,o_4,o_2,o_1)$.
Note that $T_1(R)=\{o_1,o_2\},T_3(R)=\{o_3\},T_4(R)=\{o_4\}$, and $T_2(R)=\{o_2\}$.
One can see that agents $1$ and $2$ cannot be better off because they receives their most preferred objects. Also, agents $3$ and $4$ prefer any object in $\{o_1,o_2\}$ over any object in $\{o_3,o_4\}$. Thus, $x$ is \emph{fairly produced}. However, $x$ is not \emph{pair efficient} as agents $3$ and $4$ can be better off by switching their allotments. 
 \hfill $\diamond$  
\medskip

\noindent \textbf{Internal fairness}: 
Let $C=\{c_1,c_2\}$, $N_{c_1}=\{1,2\},N_{c_2}=\{3\}$, $O_{c_1}=\{o_1,o_2\},O_{c_2}=\{o_3\}$, and $\succ_{c_1}:1,2$. 

Let $R$ be such that 
$$R_1:o_3,o_2,o_1;$$
$$R_2:o_3,o_2,o_1;$$
$$R_3:o_1,o_2,o_3.$$

Consider $x\equiv (o_2,o_3,o_1)$. $x$ is not \emph{internally fair} as $1$ envies $2$ and $1\succ_{c_1} 2$.
However, $x$ is \emph{fairly produced}. To see it, first note
that $T_1(x)=\{o_1,o_2\},T_3=\{o_3\}$, and $T_2(x)$ is either $\{o_1\}$ or $\{o_2\}$. 
So, while agent $1$ prefers $o_2$ to $o_1$, $1$ cannot be better off solely via consuming another object from $T_1(x)$.
Also note that only agent $1$ does not receive his most preferred object. It means that agent $1$ cannot be better off by trading with other agents. 
 \hfill $\diamond$   
\medskip

\noindent \textbf{External fairness}: 
Consider a modification of Serial-QR mechanism in Example~\ref{example:sqr}, where $Y(0)$ is adjusted to include all fairly produced allocations. This mechanism is \emph{procedurally fair} but is not \emph{externally fair}. \hfill $\diamond$

\subsection*{Theorem~\ref{thm:ir}}
Examples~\ref{example:vsd},~\ref{example:sp},~\ref{example:DA}, and~\ref{example:RSD} show the independence of Theorem~\ref{thm:ir}.

\end{document}